\documentclass[letterpaper]{article}
\usepackage[square,numbers]{natbib}
\usepackage{physics}
\usepackage{amsmath,amsthm,amsfonts,amssymb}
\newtheorem{theorem}{Theorem}

\newtheorem{definition}{Definition}
\newtheorem{lemma}{Lemma}
\usepackage{geometry}
\usepackage[utf8]{inputenc} 
\usepackage[T1]{fontenc}    
\usepackage{hyperref}       
\usepackage{url}            
\usepackage{booktabs}       
\usepackage{amsfonts}       
\usepackage{nicefrac}       
\usepackage{microtype}      

\makeatletter
\newcommand\ackname{Acknowledgements}
\if@titlepage
   \newenvironment{acknowledgements}{%
       \titlepage
       \null\vfil
       \@beginparpenalty\@lowpenalty
       \begin{center}%
         \bfseries \ackname
         \@endparpenalty\@M
       \end{center}}%
      {\par\vfil\null\endtitlepage}
\else
   \newenvironment{acknowledgements}{%
       \if@twocolumn
         \section*{\abstractname}%
       \else
         \small
         \begin{center}%
           {\bfseries \ackname\vspace{-.5em}\vspace{\z@}}%
         \end{center}%
         \quotation
       \fi}
       {\if@twocolumn\else\endquotation\fi}
\fi
\makeatother

\title{$2$-Local Hamiltonian with Low Complexity is QCMA complete}

\author{%
  Ying-hao~Chen\\
  \texttt{yhchen@cs.utexas.edu} \\
}

\begin{document}

\maketitle

\begin{abstract}
We prove that $2$-Local Hamiltonian ($2$-LH) with Low Complexity problem is QCMA-complete by combining the results from the QMA-completeness\cite{kempe2006complexity} of $2$-LH and QCMA-completeness of $3$-LH with Low Complexity\cite{wocjan2003two}. 
The idea is straightforward. It has been known that $2$-LH is QMA-complete. By putting a low complexity constraint on the input state, we make the problem QCMA. Finally, we use similar arguments as in \cite{kempe2006complexity} to show that all QCMA problems can be reduced to our proposed problem. 
\end{abstract}

\begin{acknowledgements}
   We thank Professor Scott Aaronson for his advices and help with this work. This work was a course project of CS 395T Quantum Complexity Theory.
\end{acknowledgements}

\section{Preliminary}
QMA is a quantum version complexity class of NP, where the verifier can be a quantum verifier and the proof is allowed to be quantum proof. QCMA is somewhat between NP and QMA. 
QCMA contains MA\cite{babai1988arthur} but contained in QMA\cite{kitaev2002classical}. 
Unfortunately, it is still open whether QCMA is strictly less powerful than QMA. 

People tried to study the difference between QCMA and QMA from many kinds of perspectives. 
One way is to study the oracle separation. 
It has been shown that there exists a quantum circuit oracle that separates QCMA and QMA\cite{aaronson2007quantum}. 
However, we still don't have any classical oracle separation between them.

We can also study their difference from the perspective of their complete problems. First $\log(n)$-LH, then $5$-LH, $3$-LH and finally $2$-LH have been proved to be QMA-complete\citep{kitaev2002classical,kempe20033,kempe2006complexity}. 
But we still don't know whether any of them is in QCMA or not, and we don't have too many QCMA-complete problems, either. 

Wocjan et al.\cite{wocjan2003two} show that adding the low complexity constraint to $3$-LH problem makes it QCMA-complete. 
We simply combine their result with the QMA-completeness of $2$-LH\cite{kempe2006complexity} to show that $2$-LH with low complexity is QCMA-complete. 
\subsection{QCMA}
We will think of QCMA, Quantum Classical Merlin-Arthur, as a class of promise problems rather than a class of languages. A promise problem $L$ can be divided into 2 disjoint sets $L = L_{yes} \cup L_{no} $ where instances are promised to be either ``Yes'' or ``No''. 
If $L\in \text{QCMA}$, there exist a quantum polynomial time verifier $V_x$ such that for any instance $x\in L$, $x$ can be verified with the help of a basic state witness $\ket{y}$ only if $x\in L_{yes} $. Formally, let $\mathcal{B} = \mathbb{C}^2$ denote the Hilbert space of a qubit. 
\begin{definition}[QCMA]
Fix $\epsilon = \epsilon(|x|) $ s.t. $2^{\Omega(|x|)} \leq \epsilon \leq \frac{1}{3}$. 
A promise problem $L = L_{yes} \cup L_{no} $ is in QCMA if for any $x\in L $, 
there exists a quantum circuit $V_x $ with $|V_x|$ elementary quantum gates which acts on the Hilbert space
\[
  \mathcal{H} := \mathcal{B}^{\otimes n_x}\otimes \mathcal{B}^{\otimes m_x}
\]
where there are $n_x $ input qubit registers and $m_x $ ancilla qubit registers and $|V_x|, n_x, m_x \in poly(|x|)$ such that
\[
  x \in L_{yes} \Rightarrow \exists y\in\{0,1\}^{n_x},\; Tr(V_x(\ket{y}\bra{y}\otimes\ket{0}\bra{0})V_x^{\dagger}P_1) \geq 1-\epsilon
\]
\[
  x \in L_{no} \Rightarrow \forall y\in\{0,1\}^{n_x},\; Tr(V_x(\ket{y}\bra{y}\otimes\ket{0}\bra{0})V_x^{\dagger}P_1) \leq \epsilon
\]
where $P_1 $ is the projection corresponding to the measurement on the first output qubit. $Tr(V_x(\ket{y}\bra{y}\otimes\ket{0}\bra{0})V_x^{\dagger}P_1)$ is the probability for the first output qubit to be state $1$ on the measurement. 
\end{definition}

\subsection{k-Local Hamiltonian problem}
$k$-LH is a quantum version of the MAX-$k$-SAT problem. 
\begin{definition}[$k$-LH]
  Given $H = \sum\limits_{i=1}^M H_i $ where $H_i $ is $k$-local. 
  That is, each $H_i $ is a Hamiltonian acting on at most $k$ qubits. It is promised that either
  \begin{enumerate}
    \item $\exists \ket{\psi} $ s.t. $\expval{H}{\psi} \leq a $, or
    \item $\forall \ket{\psi} $, $\expval{H}{\psi} \geq b $.
  \end{enumerate}
  where $ 0 < a < b < 1 $ are constants. The problem is to decide which. 
\end{definition}

\begin{definition}[Low Complexity State]
  Let $\mathcal{L}_C $ denote the set of low complexity states. 
  We say that $\ket{\psi}\in\mathcal{L}_C$ if and only if we can prepare $\ket{\psi} $ 
  by a sequence of elementary quantum gates with size polynomial in the size of $\ket{\psi} $. That is,
  \[
  \ket{\psi} = U_TU_{T-1}\cdots U_1\ket{0}
  \]
  for some elementary quantum gates $U_1,\ldots, U_T$ where $T = poly(size(\ket{\psi})) $.
\end{definition}

\section{Main Result}
In this section, we prove our main result.
\begin{theorem}[$2$-LH with Low Complexity is QCMA-complete]
  Given any $2$-local Hamiltonian $H = \sum\limits_{i=1}^M H_i $, and promised that either
  \begin{enumerate}
    \item There exists a low energy and low complexity state $\ket{y'}\in\mathcal{L}_C $ s.t. 
    \[
      \expval{H}{y'} \leq \epsilon
    \]
    \item or for any low complexity states $\ket{y'}\in\mathcal{L}_C$, 
    \[
      \expval{H}{y'} \geq \frac{1}{2}-\epsilon
    \]
  \end{enumerate}
  The $2$-Local Hamiltonian with Low Complexity ($2$-LHLC) problem is to decide which. 
  $2$-LHLC is QCMA-complete.
\end{theorem}
\begin{proof}

\textbf{(Contained in QCMA)}
First note that the restriction to low complexity states makes every $k$-LHLC problem a QCMA problem. If the Hamiltonian $H$ has a low complexity low energy state $\ket{y'} $, we can use a classical proof to encode how to prepare such a state. It can be verified in quantum polynomial time. 

The only thing we need to check is that after applying the low complexity constraint, all QCMA problems can still be reduced to the $2$-LHLC problem.
\paragraph{\textbf{(completeness of $k$-LHLC)}}
We start from $k$-LHLC problem with $k = O(\log n)$, where $n = |x|$. Given a QCMA problem $L$, by definition, for each instance $x\in L$, there exists a quantum circuits $V_x = U_TU_{T-1}\cdots U_1 $ such that 
\begin{enumerate}
  \item $x\in L_{yes} \Rightarrow \exists\ket{y}\in\{0,1\}^{n_x} $ such that 
    \[
    \mathbb{P}[\text{get $\ket{0}$ on the first qubit of } V_x(\ket{y}\ket{0})\text{ after measurement}] \leq \epsilon
    \]
  \item $x\in L_{no} \Rightarrow \forall\ket{y}\in\{0,1\}^{n_x} $,
    \[
    \mathbb{P}[\text{get $\ket{0}$ on the first qubit of } V_x(\ket{y}\ket{0})\text{ after measurement}] \geq 1-\epsilon
    \]
\end{enumerate}
For this instance, we can construct a $k$-local Hamiltonian\cite{kempe2006complexity} by Kempe's construction 
for $(n_x+m_x+\log T )$-qubits low energy low complexity states.
\[
  H = J_{in}H_{in} + H_{out} +J_{prop}\sum\limits_{t=1}^T{ }H_{prop,t}
\]
\[
  H_{in} = \sum\limits_{i=n_x+1}^{n_x+m_x}\ketbra{1}_i \otimes \ketbra{0}
\]
\[
  H_{out} = (T+1)\ketbra{0}_1 \otimes \ketbra{T}
\]
\[
  H_{prop,t} = \frac{1}{2}(I\otimes \ketbra{t} + I\otimes \ketbra{t-1} - U_t\otimes \ket{t}\bra{t-1} - U_t^\dagger\otimes\ket{t-1}\bra{t})
\]
\[
  \ket{y'} = \frac{1}{\sqrt{T+1}}\sum\limits_{t=0}^T U_tU_{t-1}\cdots U_1(\ket{y}\ket{0})\otimes\ket{t}
\]

The first part is $(n_x + m_x) $ computational qubits and the second part is $\log(T+1)$ clock qubits. 
These hamiltonians are $O(\log(n))$-local because they acts on at most $2+\log(T+1)$ qubits at a time. 

Clearly, if $x\in L_{yes} $, then $\ket{y'} \in \mathcal{L}_C $ is low complexity because we only need polynomial number of quantum gates to prepare the basic state $(\ket{y}\ket{0})$ and it takes at most $poly(T)$ quantum gates operating on them to get $\ket{y'} $.
Moreover, 
\[
  \expval{H}{y'} = \expval{H_{out}}{y'} = (V_x(\ket{y}\ket{0}))^\dagger\ketbra{0}_1 V_x(\ket{y}\ket{0}) \leq \epsilon
\]

\paragraph{\textbf{(soundness of $k$-LHLC)}}
For the soundness, we need the projection lemma in \cite{kempe2006complexity},
\begin{lemma}[Projection Lemma](Please refer to the proof in \cite{kempe2006complexity})
  Given two hamiltonians $H_1, H_2 $. Let $\mathcal{S}_2 $ be the zero eigen space of $H_2$ and the eigenvectors in $\mathcal{S}_2^\perp $ has eigenvalue at least $J > 2||H_1||$. Then,
  \[
    \lambda(H_1|_{\mathcal{S}_2}) - \frac{||H_1||^2}{J-2||H_1||} \leq \lambda(H_1 + H_2) \leq \lambda(H_1|_{\mathcal{S}_2})
  \]
  where $\lambda(\cdot) $ denote the smallest eigenvalue and $\lambda(H_1|_{\mathcal{S}_2})$ is the smallest eigenvalue of $H_1$ corresponding to all eigenvectors orthogonal to $\mathcal{S}_2^{\perp} $. Moreover, we can choose $J_2 $ large enough so that
  \[
    JJ_2 > 2||H_1|| + 8 ||H_1||^2
  \]
  and hence
  \[
    \lambda(H_1|_{\mathcal{S}_2}) - \frac{1}{8} \leq \lambda(H_1 + J_2H_2)
  \]
\end{lemma}
\paragraph{}
Let $\mathcal{S}_{in} $ denote the zero eigenspace of $H_{in} $. We can see that the space is actually a space for valid inputs. That is, a $n_x $ qubits states follows by $m_x $ ancilla qubits. With projection lemma, 
we can lower bound $\lambda(H)$ by 
\[
  \lambda((H_{out} + J_{prop}\sum\limits_{t=1}^T H_{prop,t}) |_{\mathcal{S}_{in}}) - \frac{1}{8} \leq \lambda(H)
\]
In other words, we can simply rule out other invalid input states by choosing $J_{in} $ large enough. We can regard $H_{in} $ and $H_{prop,t} $ as constraints that force the state to do exactly what we want. For example, valid input and states going through $U_T, \ldots, U_1 $. If any of them violated, it would cause large energy to $H$.

Similar to $H_{in} $, we can also choose $J_{prop} $ large enough so that 
\[
  \lambda(H_{out}|_{\mathcal{S}_{in} \cap \mathcal{S}_{prop}}) - \frac{1}{8} \leq \lambda((H_{out} + J_{prop}\sum\limits_{t=1}^T H_{prop,t}) |_{\mathcal{S}_{in}})
\]
and hence
\[
  \lambda(H_{out}|_{\mathcal{S}_{in} \cap \mathcal{S}_{prop}}) - \frac{2}{8} \leq \lambda(H)
\]
By the definition of QCMA, if $x\in L_{no} $, then
\[
  \lambda(H_{out}|_{\mathcal{S}_{in} \cap \mathcal{S}_{prop}}) \geq 1-\epsilon
\]
which results in 
\[
  \lambda(H) \geq \frac{3}{4}- \epsilon
\]
Note that we should write $\lambda(H_{out}|_{\mathcal{L}_C \cap \mathcal{S}_{in} \cap \mathcal{S}_{prop}}) - \frac{2}{8} \leq \lambda(H|_{\mathcal{L}_C})$. For simplicity, we ignore the notation for low complexity constraint $\mathcal{L}_C$ while writing $\lambda(\cdot) $. Therefore, $k$-LHLC with $k = O(\log n) $ is QCMA-complete. 

\paragraph{\textbf{(From $k$-LHLC to $2$-LHLC)}} Note that if we use unary representation to keep the clock qubits, $k$ can be reduce to $5$. We can replace $\ket{t}\bra{t-1} $ by $\ket{110}\bra{100}_{t-1} $ and hence the bottleneck would be the term
\[
  U_t \otimes \ket{110}\bra{100}_{t-1} \text{ , and } U_t^\dagger \ket{100}\bra{110}_{t-1}
\]
They operate on at most $5$ qubits. $2$ for computational qubits in $U_t $ and $3$ for clock qubits.

It has also been shown that, actually, we just need $1$ qubit to keep the clock if we can somehow ensure the clock to be always valid. i.e. $1$ always happens before $0$. It turns out that we can simply add more clock constraints to $H$ to ensure this. Let
\[
  H = J_{in}H_{in} + H_{out} +J_{prop}\sum\limits_{t=1}^T{ }H_{prop,t} + J_{clock}H_{clock}
\]
where $H_{clock} = \sum\limits_{1\leq i < j \leq T} \ketbra{01}_{ij} $ and we use $T$ qubits to keep the clock as uniry representation. Other parts remain the same. Then, the bottleneck becomes
\[
  U_t \otimes \ket{1}\bra{0}_{t} \text{ , and } U_t^\dagger \ket{0}\bra{1}_{t}
\]
which operates on at most $3$ qubits now. Note that all other hamiltonians are already $2$-local. Let 
\[
  C_\phi = \begin{pmatrix}
    1 & 0 & 0 & 0  \\
    0 & 1 & 0 & 0  \\
    0 & 0 & 1 & 0  \\
    0 & 0 & 0 & -1 \\
  \end{pmatrix}, 
  Z = \begin{pmatrix}
  1 & 0 \\
  0 & -1\\
  \end{pmatrix}
\]
One fact is that $C_\phi $ and all 1-qubits gates is universal. WLOG, we only need to focus on the hamiltonian with $U_t = C_\phi $. Moreover, $C_\phi = (Z\otimes I)(I\otimes Z)C_\phi(I\otimes Z)(Z\otimes I) $. We can replace all $C_\phi $ by these $5$ gates in sequence. We can also add any $I$ gates to $V_x$ in order to make sure that $C_\phi $ locates at time $L, 2L, \ldots, T_{2}L$. Since the $2$-qubit gate is $C_\phi $ and it only situates at $L, 2L,\ldots $ and follows and leads by $2$ $Z$-gates. We can check the propagation of states by directly pairwise compare the states in $\ket{L-2}, \ket{L-1}, \ket{L}, \ket{L+1}, \ket{L+2}, \ket{L+3} $ without using $C_\phi \otimes \ket{L} $ as constraint. It reduces the $3$-LHLC to $2$-LHLC. The final version of Kempe's construction of local hamiltonians is as follows:
\[
  H = H_{out}+J_{in}H_{in} +J_{1}H_{prop1}+J_{2}H_{prop2} + J_{clock}H_{clock}  
\]
\[
  H_{in} = \sum\limits_{i = n_x+1}^{n_x+m_x}\ketbra{1}_i\otimes \ketbra{0}_1,\;\;
  H_{out} = (T+1)\ketbra{0}_1\otimes \ketbra{1}_T
\]
\[
  H_{clock} = \sum\limits_{1\leq i< j \leq T} \ketbra{01}_{ij}
\]
\[  
  H_{prop1} = \sum\limits_{t\in T_1} H_{prop,t}, \;\; T_1 = \{1,\ldots T\}\backslash \{L, 2L,\ldots\}
\]
\[
  H_{prop,t} = \frac{1}{2}(I\otimes \ketbra{10}_{t,t+1} + I\otimes\ketbra{10}_{t-1,t} - U_t\otimes\ket{1}\bra{0}_t - U_t^\dagger\otimes\ket{0}\bra{1}_t)
\]
\[
  H_{prop,1} = \frac{1}{2}(I\otimes \ketbra{10}_{1,2} + I\otimes\ketbra{0}_{1} - U_t\otimes\ket{1}\bra{0}_1 - U_t^\dagger\otimes\ket{0}\bra{1}_1)
\]
\[
  H_{prop,T} = \frac{1}{2}(I\otimes \ketbra{1}_{T} + I\otimes\ketbra{10}_{T-1,T} - U_t\otimes\ket{1}\bra{0}_T - U_t^\dagger\otimes\ket{0}\bra{1}_T)
\]
\[
  H_{prop2} = \sum\limits_{\ell =1}^{T_2}(H_{qubit,\ell L} + H_{time, \ell L})
\]
and with $f_t $ and $s_t $ being the first and second qubits of $C_\phi $ gate at time t, 
\[
  H_{qubit,t} = \frac{1}{2}(-2\ketbra{0}_{f_t} - 2\ketbra{0}_{s_t} + \ketbra{1}_{f_t} + \ketbra{1}_{s_t}) \otimes (\ket{1}\bra{0}_t \otimes \ket{0}\bra{1}_t)
\]
\begin{eqnarray*}
\begin{aligned}
  &H_{time, t} = \frac{1}{8}I\otimes &(&\ketbra{10}_{t,t+1} + 6\ketbra{10}_{t+1,t+2} + \ketbra{10}_{t+2,t+3}&&\\
  &&+&2\ket{11}\bra{00}_{t+1,t+2} + 2\ket{00}\bra{11}_{t+1,t+2} &\\
  &&+&\ket{1}\bra{0}_{t+1} + \ket{0}\bra{1}_{t+1} + \ket{1}\bra{0}_{t+2} + \ket{0}\bra{1}_{t+2} &\\
  &&+&\ketbra{10}_{t,t-1} + 6\ketbra{10}_{t-1,t-2} + \ketbra{10}_{t-2,t-3}&&\\
  &&+&2\ket{11}\bra{00}_{t-1,t-2} + 2\ket{00}\bra{11}_{t-1,t-2} &\\
  &&+&\ket{1}\bra{0}_{t-1} + \ket{0}\bra{1}_{t-1} + \ket{1}\bra{0}_{t-2} + \ket{0}\bra{1}_{t-2} )&\\
\end{aligned}
\end{eqnarray*}
The completeness is straightforward. If $x\in L_{yes} $, we can construct $\ket{y'} $ as in previous $k$-LHLC. We will get
\[
  \expval{H}{y'} = \expval{H_{out}}{y'} \leq \epsilon
\]
The soundness is proved by repeatedly applying the projection lemma.
\begin{eqnarray*}
\begin{aligned}
  &&&\lambda(H)&\\
  &\geq&&\lambda(H_{out} + J_{in}H_{in} + J_{prop2}H_{prop2} + J_{prop1}H_{prop1}|_{\mathcal{S}_{clock}}) - \frac{1}{8}&\\
  &\geq&&\lambda(H_{out} + J_{in}H_{in} + J_{prop2}H_{prop2}|_{\mathcal{S}_{clock} \cap S_{prop1}}) - \frac{2}{8}&\\  
  &&\vdots&&\\
  &\geq&&\lambda(H_{out}|_{\mathcal{S}_{clock} \cap\mathcal{S}_{prop1}\cap\mathcal{S}_{prop} \cap \mathcal{S}_{in}})-\frac{4}{8}&\\
  &\geq&&1-\epsilon-\frac{4}{8}&\\
  &=&&\frac{1}{2}-\epsilon&\\
\end{aligned}
\end{eqnarray*}
The elimination of $J_{prop2}H_{prop2} $ is not exactly the same as with other hamiltonians, but the results are similar. 
For more details, please refer to \cite{kempe2006complexity}. The only difference between our $H$ and theirs is that throughout the whole argument, we restrict the input state to be low complexity $\mathcal{L}_C $. Even while we refer to the smallest eigenvalue of some hamiltonian $\lambda(H) $, we only refer to those corresponding to low complexity eigenvectors. Most of the results inherit directly from the original $2$-LH construction. 
\end{proof}

\bibliography{bibliography}
\end{document}